\documentclass[letterpaper,twocolumn,10pt]{article}
\usepackage{usenix-2020-09}
\usepackage{listings}
\microtypecontext{spacing=nonfrench}
\usepackage{tikz}
\usetikzlibrary{shapes,arrows.meta, positioning,fit}
\usepackage{float}
\usepackage{amsmath, amssymb, amsthm}
\usepackage{caption}
\usepackage[linesnumbered, boxed, ruled]{algorithm2e}
\usepackage{indentfirst}
\usepackage{graphicx}
\usepackage{float}
\usepackage{booktabs}
\graphicspath{ {./images/} }

\newcommand{\func}[1]{\ensuremath{\mathcal{F}_{#1}}}
\newcommand{\prot}[1]{\ensuremath{\Pi_{#1}}}
\newcommand{\adv}{\ensuremath{\mathcal{A}\,\,}}
\newcommand{\simul}{\ensuremath{\mathcal{S}\,\,}}
\newcommand{\env}{\ensuremath{\mathcal{Z}}}
\newcommand{\leak}{\ensuremath{\mathcal{L}}}
\newcommand{\proc}[1]{\text{\normalfont\scshape{#1}}}

\SetAlCapSkip{1em}

\begin{document}

\date{}

\title{\Large \bf FHE-SQL: Fully Homomorphic Encrypted SQL Database}

\author{
    {\rm Po-Yu Tseng}\\
    National Taiwan University
    \and
    {\rm Po-Chu Hsu}\\
    FHEngine.com
    \and
    {\rm Shih-Wei Liao}\\
    National Taiwan University
}

\maketitle
\def\KeyGen{\mathsf{KeyGen}}
\def\Encrypt{\mathsf{Encrypt}}
\def\Decrypt{\mathsf{Decrypt}}
\def\Eval{\mathsf{Eval}}
\section*{Abstract}

FHE-SQL is a privacy-preserving database system that enables secure query processing on encrypted data using Fully Homomorphic Encryption (FHE), providing privacy guaranties where an untrusted server can execute encrypted queries without learning either the query contents or the underlying data. Unlike property-preserving encryption-based systems such as CryptDB, which rely on deterministic or order-preserving encryption and are vulnerable to frequency, order, and equality-pattern inference attacks, FHE-SQL performs computations entirely under encryption, eliminating these leakage channels. Compared to trusted-hardware approaches such as TrustedDB, which depend on a hardware security module and thus inherit its trust and side-channel limitations, our design achieves end-to-end cryptographic protection without requiring trusted execution environments. In contrast to high-performance FHE-based engines-Hermes, which target specialized workloads such as vector search, FHE-SQL supports general SQL query semantics with schema-aware, type-safe definitions suitable for relational data management. FHE-SQL mitigates the high cost of ciphertext space by using an indirection architecture that separates metadata in RocksDB from large ciphertexts in blob storage. It supports oblivious selection via homomorphic boolean masks, multi-tier caching, and garbage collection, with security proven under the Universal Composability framework.

\vspace{0.2em}
\noindent\textbf{Keywords:} Fully Homomorphic Encryption, Private Information Retrieval, Encrypted Databases, Privacy-Preserving Systems
\section{Introduction}
\label{sec:introduction}

The proliferation of cloud computing has created an unprecedented demand for outsourcing data storage and processing to untrusted third-party servers. Organizations increasingly rely on cloud database services for their scalability, cost-effectiveness, and maintenance benefits. However, this trend introduces significant privacy concerns: sensitive data must be stored and processed on servers controlled by potentially adversarial entities. Traditional approaches to this problem have relied on property-preserving encryption schemes such as deterministic encryption\cite{GenericAttacksonSecureOutsourcedDatabases}, order-preserving encryption, and format-preserving encryption. Systems such as CryptDB\cite{CryptDB} pioneered the use of these techniques to enable SQL queries over encrypted data. However, these approaches have fundamental security limitations: they leak substantial information about the underlying data through access patterns, query equality patterns, and statistical properties.

Recent cryptanalytic attacks\cite{PPE_Attacks} have shown that property-preserving encryption schemes can be completely broken in many practical scenarios. Deterministic encryption reveals which plaintexts are equal, order-preserving encryption leaks the relative ordering of all encrypted values, and searchable encryption exposes search patterns that can be exploited to recover plaintext data. We propose FHE-SQL that achieves privacy information retrieval against the above-mentioned attacks by using Fast Fully Homomorphic Encryption over the Torus \cite{cryptoeprint:2018/421} as the foundation for privacy-preserving database systems. TFHE allows arbitrary computations to be performed on encrypted data without decryption, providing the strongest possible privacy guarantees, short of information-theoretic security. We give a security analysis
of our system in Section \ref{sec:security}.

Early attempts\cite{gahi2015securedatabaseusinghomomorphic} to apply FHE to secure SQL databases faced practical challenges, as performing SQL-style queries on encrypted data proved difficult\cite{mani2013enablingsecuredatabaseservice}. Applying FHE to database systems remains a significant technical challenge. The reasons are as follows. First, homomorphic operations incur significant computational overhead, especially in executing complex queries that require multiple levels of circuit evaluation\cite{xue2025measuringcomputationaluniversalityfully}. Second, most current FHE schemes are optimized for basic arithmetic and boolean operations; when extending to more sophisticated computations, the execution time can grow exponentially. This limitation also affects query processing, as indexing of encrypted columns is infeasible, forcing the system to adopt an iterative scan approach. Finally, FHE ciphertexts are substantially larger than their plaintext counterparts—often ranging from hundreds of kilobytes to several megabytes per value.

Our work presents FHE-SQL, a database system that executes SQL queries directly over encrypted data without relying on property- or order-preserving schemes, thus avoiding the vulnerabilities inherent in those methods. This enables applications that need privacy guaranties, such as analytics on confidential patient records in a third-party cloud or querying sensitive financial data for fraud detection without exposing the raw information. By leveraging FHE, the system ensures strong privacy guarantees against untrusted servers, preserving the confidentiality of stored data, query contents, and query results. In contrast to property-preserving encryption–based solutions, FHE-SQL offers formal security proofs against inference attacks that exploit data and access patterns. The key contributions of this work are as follows and we provided implementation as a zip file in supplemental material.

\begin{itemize}
    \item \textbf{Oblivious Query Protocol:} We introduced a protocol for private data retrieval using homomorphic boolean mask generation. This technique allows the server to execute selection queries obliviously, ensuring it learns nothing about the query's filter conditions or the resulting dataset.
    \item \textbf{Performance Enhancement:} We integrate computation-aware query planning and batching strategies to reduce the overhead of homomorphic operations, achieving improved latency and throughput in practical workloads. By leveraging blob storage instead of naively storing all ciphertexts in RocksDB, we achieved up to 12× faster concurrent writes and 45× faster concurrent reads.
    \item \textbf{Security Analysis:} We provide formal security analysis proving that our system achieves computational privacy against adversaries, with privacy guarantees that degrade gracefully under various attack models within the universal composability (UC) framework.
    \item \textbf{Extensibility:} The system architecture is designed to support complex database schemas, multiple query types, and future integration with advanced FHE schemes.
\end{itemize}
\section{Preliminary}
\label{sec:preliminary}
This section introduces the concept and security model of Fully Homomorphic Encryption (FHE), outlines the main schemes used in practice, and explains how FHE can be leveraged to construct Private Information Retrieval (PIR) protocols.

\subsection{Private Information Retrieval}
\noindent A Private Information Retrieval (PIR) scheme enables a user to retrieve a specific data item from a database without revealing the identity of the item to the server holding the data. While this concept is foundational for building privacy-preserving databases, its practical implementation has been limited by the high computational cost to servers. The computational solution to PIR was proposed in 1997\cite{PIR}, and several schemes, such as SealPIR\cite{cryptoeprint:2017/1142} and DEPIR\cite{cryptoeprint:2022/1703}, came up to reduce the computational cost and communication time. A PIR scheme must satisfy two core requirements to be considered effective and secure:

\begin{itemize}
    \item \textbf{Correctness}: The user must be able to accurately reconstruct the requested message from the responses received from the server(s).
    \item \textbf{Privacy}: The index of the item a user is retrieving must remain confidential from the server(s). This privacy is often an information-theoretic guarantee, meaning it holds true even against a computationally powerful adversary who could analyze the queries.
\end{itemize}

\noindent In essence, our FHE-SQL achieves PIR by transforming database queries into homomorphic computations. This protects the privacy of the user's requests by obscuring the content of the query from the untrusted server.

\subsection{Fully Homomorphic Encryption}
\noindent The concept of Fully Homomorphic Encryption, first proposed by Rivest in 1978 and theoretically realized by Gentry\cite{GentryFHE}, enables arbitrary computations on encrypted data without requiring decryption. Modern FHE schemes are typically based on the Learning With Errors (LWE) problem or its ring variant (Ring-LWE). A fully homomorphic encryption scheme consists of algorithms ($\KeyGen$, $\Encrypt$, $\Decrypt$, $\Eval$) where:

\vspace{0.3em}
\begin{itemize}
    \item $\KeyGen(1^\lambda) \rightarrow (pk, sk, evk):$
          Generates public key, secret key, and evaluation key
    \item $\Encrypt(pk | sk, m) \rightarrow c:$
          Encrypts message m under public key $pk$ or secret key $sk$.
    \item $\Decrypt(sk, c) \rightarrow m:$
          Decrypts ciphertext $c$ using secret key $sk$
    \item $\Eval(evk, f, c_1,...,c_k) \rightarrow c':$
          Homomorphically evaluates function f on ciphertexts by under evaluation key
\end{itemize}

A scheme is semantically secure if no polynomial-time adversary can distinguish between encryptions of two chosen messages, and required to pass the modern security model IND-CPA-D\cite{IND-CPA-D} which is a security model aiming for homomorphic encryption. 

TFHE\cite{cryptoeprint:2018/421} and CKKS\cite{cryptoeprint:2016/421} are two primary FHE schemes that are widely recognized and supported by several mature libraries. TFHE enables the construction of universal logic gates and circuits and provides efficient implementations of logical operations. This scheme has been proven secure under the IND-CPA-D model\cite{TFHE_IND-CPA-D}. CKKS supports SIMD operations and is particularly efficient for machine learning applications. However, it is designed for approximate arithmetic on real numbers and has a limited circuit depth, which is not suitable for our work.

We adopt TFHE due to its efficiency for boolean operations, such as equality comparisons and the generation of boolean masks which are central to SQL query processing. Our implementation leverages the \texttt{TFHE-rs\cite{TFHE-rs}} library to achieve practical performance while maintaining strong security guarantees.

\subsection{FHE-Based PIR}
\noindent FHE provides a natural approach to PIR construction because of the homomorphic capability:

\begin{itemize}
    \item Client encrypts query index: $c_i = \Encrypt(pk, i)$
    \item Server computes: $c_{result} = \Eval(evk, \text{eq}, c_i, j) \cdot \text{DB}[j]$
    \item Client decrypts $c_{result}$ to recover $\text{DB}[i]$
    \item Where $\text{eq}(x,y)$ is a homomorphic equality predicate returning $1$ if $x=y$, $0$ otherwise.
\end{itemize}

\noindent Our FHE-SQL system extends this approach to support complex database schemas and multiple query types while optimizing for practical performance.
\begin{figure*}[!htb] 
  \centering
  \includegraphics*[width=\textwidth]{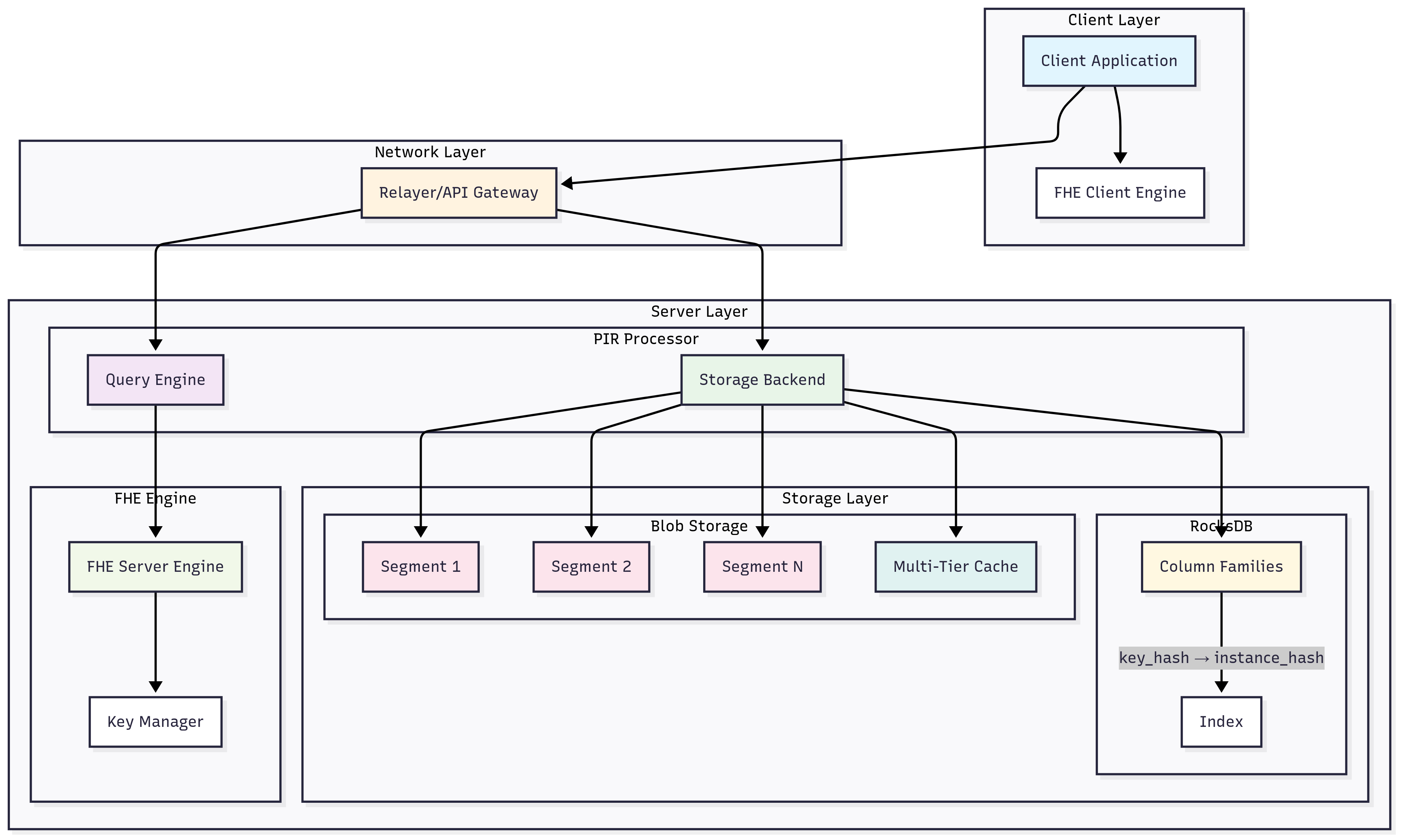}
  \caption{Overview of FHE-SQL system architecture.} 
  \label{fig:fhe-sql} 
\end{figure*}
\section{Related Works and Cryptanalytic Attacks}
\label{sec:related_works}

A wide range of approaches have been proposed to enable secure computation over encrypted databases, ranging from traditional property-preserving encryption techniques to recent systems based on fully homomorphic encryption (FHE). In this section, we review some works in both categories, highlighting their supported operations, performance characteristics, and security guarantees. Also, we summarize notable cryptanalytic attacks that have revealed fundamental vulnerabilities in non-FHE approaches, underscoring the motivation for adopting stronger cryptographic primitives in our design.

\subsection{Non-FHE Approaches}
\noindent Early practical encrypted database systems relied heavily on property-preserving encryption or trusted hardware to enable query processing without full decryption.
These methods typically trade security for functionality and performance, as shown in the following systems:
\begin{itemize}
    \item \textbf{CryptDB\cite{CryptDB}:} The first practical encrypted database system, using multiple encryption schemes (deterministic, order-preserving, searchable) to support different SQL operations. However, CryptDB has significant security limitations due to the information leaked by property-preserving encryption.
    \item \textbf{Always Encrypted\cite{AlwaysEncrypted}:} Commercial implementation providing deterministic and randomized encryption with application-layer decryption. Limits supported operations to equality queries for deterministic encryption.
    \item \textbf{TrustedDB\cite{TrustedDB}:} Uses a trusted hardware component to process sensitive operations while keeping most data encrypted in untrusted storage.
\end{itemize}

\subsection{FHE-Based Approaches}
\noindent Fully Homomorphic Encryption (FHE) eliminates the need for property-preserving encryption by enabling arbitrary computation over encrypted data. Although early FHE-based database systems were limited by performance constraints, they have progressively expanded supported operations and application domains. Representative systems include:

\begin{itemize}
    \item \textbf{Cipherbase\cite{Cipherbase}:} Early prototype using somewhat homomorphic encryption for aggregate queries. Limited to simple operations due to noise accumulation.
    \item \textbf{Seabed\cite{Seabed}:} Used FHE for specific analytics workloads but did not address general-purpose database operations or storage optimization.
    \item \textbf{HEDB\cite{HEDB}:} Implemented basic database operations using CKKS FHE but focused on numerical computations rather than general database queries.
    \item \textbf{Hermes\cite{Hermes}:} Integrated FHE-native vector query processing into a standard SQL engine. It leverages multi-slot FHE packing in OpenFHE to embed multiple records and auxiliary statistics per ciphertext.
\end{itemize}

\noindent Although these works demonstrate the feasibility of FHE for encrypted databases, most remain specialized to certain workloads or face performance bottlenecks in practical deployments.

\subsection{Cryptanalytic Attacks}
\noindent Beyond functional limitations, many non-FHE approaches have been shown to suffer from severe cryptanalytic vulnerabilities.
In particular, property-preserving encryption leaks structural information that can be exploited through the following attacks:
\begin{itemize}
    \item \textbf{Frequency Analysis Attacks\cite{FrequencyAttacks}:} Deterministic encryption reveals equality patterns, enabling frequency analysis attacks that can recover plaintext data when the plaintext distribution is known.
    \item \textbf{Order-Revealing Attacks\cite{OrderPreservingAttacks}:} Order-preserving encryption exposes the complete ordering of encrypted values, allowing reconstruction of approximate plaintext values.
    \item \textbf{Inference Attacks\cite{AccessPatternAttacks}:} Access pattern analysis can reveal significant information about query contents and results, even with strong encryption schemes.
\end{itemize}

\noindent These attacks demonstrate the fundamental insecurity of property-preserving approaches and motivate the need for stronger cryptographic techniques like FHE.
\section{FHE-SQL: System Design and Architecture}
\label{sec:fhe_sql}

In this section, we present an in-depth examination of our FHE-SQL system, which enables secure execution of SQL queries directly over encrypted data. This design supports expressive and complex queries without exposing plaintext to the server at any stage. An overview of the system architecture is shown in Figure~\ref{fig:fhe-sql}.
\begin{algorithm*}[ht]
    \caption{FHE-SQL Client-Side Query Execution}
    \label{alg:fhe_sql_client_2e}
    
    \SetKwInOut{KwIn}{Require}
    \SetKwInOut{KwOut}{Ensure}
    
    \KwIn{SQL Query $q$, FHE key-pair $(sk, pk)$, Signing key $sk_{sig}$}
    \KwOut{Plaintext result set $R_{plain}$}
    
    $Q \leftarrow \proc{EncryptQuery}(q)$\ \tcp*{Encrypt values in a SQL query}
    $Token \leftarrow \proc{CreateAuthToken}(\text{owner\_id, \dots}, sk_{sig}, \text{Nonce})$\ \tcp*{ Data owner creates a auth token} 
    
    Send $(Q, Token)$ to Server \& Receive encrypted result $C_{result}$. \tcp*{ Communication with Server}
    
    $R_{masked} \leftarrow \proc{Decrypt}(C_{result}, sk)$\ \tcp*{Decrypt the entire result set}
    $R_{plain} \leftarrow \proc{FilterZeroRows}(R_{masked})$\ \tcp*{Remove rows that evaluated to false}
    \Return{$R_{plain}$}\;
\end{algorithm*}
\begin{algorithm*}[t]
    \SetKwInOut{KwIn}{Require}
    \SetKwProg{Proc}{Procedure}{}{}
    \SetKwProg{Fn}{Function}{}{}
    \SetKwFunction{FnHET}{EvaluateHomomorphicTree}
    
    \caption{FHE-SQL Server-Side Query Execution}
    \label{alg:fhe_sql_server_2e}

    \KwIn{Encrypted Query: $Q$, Access Token: $Token$, Verification Key: $pk_{sig}$}

    \Proc{ProcessQuery($Q, Token$)}{
        \If{\textbf{not} \proc{VerifySignature}($Token, pk_{sig}$)}{
            \Return Error\;
        }
        $Q_{ast} \leftarrow \proc{ParseQueryIntoAst}(Q)$ \tcp*{Parse encrypted query $Q$ into AST}
        $DB_{meta} \leftarrow \proc{FetchMetadata}(Q_{ast}.\text{from})$ \tcp*{Get metadata from RocksDB}
        $DB \leftarrow \proc{FetchCiphertexts}(DB_{meta})$ \tcp*{Load encrypted data from blob storage}
        
        Let $M$ be an empty vector of ciphertexts \tcp*{Initialize boolean selection mask}
        \ForEach{encrypted row $r \in DB$}{
            $m \leftarrow \FnHET(Q_{ast}.\text{where}, r)$ \tcp*{Generate mask bit via FHE-PIR}
            Append $m$ to $M$\;
        }
        
        $C_{selected} \leftarrow \proc{SelectEncryptedColumns}(DB, Q_{ast}.\text{columns})$\;
        $C_{result} \leftarrow \proc{Homomorphic If-Then-Else}(C_{selected}, M)$ \tcp*{Apply mask to selected data}
        \Return $C_{result}$ to Client\;
    }
    \BlankLine
    \Fn{\FnHET{node, row}}{
        \uIf{node is Identifier}{
            \Return $row[\text{node.name}]$\;
        }
        \uElseIf{node is EncryptedLiteral}{
            \Return $\text{node.value}$\;
        }
        
        $L \leftarrow \FnHET(\text{node.left, row})$\;
        $R \leftarrow \FnHET(\text{node.right, row})$\;
        
        \Return \proc{ApplyHomomorphicOp}(\text{node.operator}, L, R)\;
    }
\end{algorithm*}

\subsection{System Architecture}
\subsubsection*{A. RocksDB and Blob Storage Integration}
\noindent FHE ciphertexts present unique storage challenges due to their size (100KB-10MB per value).
Traditional database storage engines are optimized for small to medium records (typically <100KB).
Therefore, we provides a hybrid architecture to combine the benefits of LSM-trees for metadata with specialized blob storage for FHE ciphertexts.

\subsubsection*{B. Challenges of LSM Trees with Large Values}
\noindent While Log-Structured Merge (LSM) trees, the core architecture of systems like RocksDB, are highly optimized for write-intensive workloads, their performance can degrade significantly when handling extremely large values. This degradation stems from several interconnected factors inherent to the LSM-tree design:

\begin{itemize}
    \item \textbf{Write Amplification:} The fundamental process of compaction in an LSM-tree requires that data be repeatedly read from one level and rewritten to the next. When values are large, the I/O cost of this process is magnified substantially. A single large key-value pair, moved through multiple compaction levels, can consume a disproportionate amount of disk bandwidth, leading to severe write amplification and increased strain on storage hardware.

    \item \textbf{Memory Pressure:} Large values exert significant pressure on in-memory resources. In the MemTable, they consume valuable space, reducing the number of individual keys that can be buffered before a flush to disk is required. More critically, in the block cache, a single large value can displace dozens or even hundreds of smaller entries. This "cache poisoning" reduces the overall effectiveness of the cache, as fewer hot keys can be retained in memory.

    \item \textbf{Read Performance Degradation:} The memory pressure caused by large values in FHE-SQL directly impacts read latency. With a less effective cache, the probability of a cache miss increases. Consequently, read operations are more likely to require fetching data from SSTables on disk---a far slower operation than reading from memory. This effect is particularly detrimental for read-intensive workloads, where low-latency access is paramount.
\end{itemize}

\subsubsection*{C. Blob Storage Systems}
To optimize the handling of large and unstructured objects, we implemented a storage system. The architecture is segment-based, where data are appended sequentially into large segments to maximize write throughput. For efficient data retrieval, we utilizes memory mapping, a technique that allows for random access to parts of a large object without copying the file into user space. To accommodate varied usage, it features tiered caching mechanisms, using multiple layers such as hot, warm, and cold caches to serve data according to different access patterns.

Our system integrates the two: RocksDb and blob storage, stores encrypted key-value pairs, where both the key and value are ciphertexts. To manage them, we generate a hash for each ciphertext. A Key hash maps a value hash. The key hashes are stored in RocksDB, and the actual FHE ciphertexts are stored in blob storage, This approach combines the high performance of Log-Structured Merge trees for managing metadata with the specialized capabilities of blob storage for the large FHE ciphertexts.

\subsection{FHE-based PIR Algorithms}
For key-value access from a database, a user creates an encrypted key which represents a query and sends it to the database with schema information attached.
The processor first filters the dataset by schema to ensure type compatibility and reduce computational overhead, then iterates through the dataset and performs homomorphic equality comparison with the encrypted key.
The system then generates a boolean mask $M_{enc}$, which is a vector of encrypted boolean values where each position corresponds to a database entry.
Entries that match the query criterion produce an encrypted representation of logical true, while non-matching entries produce an encrypted representation of logical false.
This boolean mask $M_{enc}$ serves as a cryptographic filter that can be applied homomorphically to the entire dataset.

After $M_{enc}$ is generated, the processor applies it to the entire dataset and returns the result to the user.
The data not matched will be encrypted zero after $M_{enc}$ is applied, otherwise it remains the same.
This approach ensures that every database entry is processed identically, preventing information leakage about query selectivity or matching patterns.
We cannot build an index for the encrypted data, since it reveals the data's comparability. See figure \ref{fig:fhe_kv_search} for the diagram.

\usetikzlibrary{
    matrix,
    positioning,
    arrows.meta,
    fit
}

\begin{figure*}[t]
    \centering
    \begin{tikzpicture}[
        node distance=12mm and 6mm,
        cell/.style={rectangle, draw, minimum width=2.2cm, minimum height=0.7cm, align=center},
        keycell/.style={cell, fill=blue!20},
        datacell/.style={cell, fill=green!20},
        querycell/.style={cell, fill=red!20, minimum width=3cm},
        maskcell/.style={cell, fill=yellow!20},
        index/.style={font=\scriptsize, inner sep=2pt},
        title/.style={font=\bfseries},
        arr/.style={->, >=Latex, thick, gray}
    ]
    
    \matrix (data) [
        matrix of nodes, row sep=1mm, column sep=2mm, nodes={keycell}
    ] {
        $\text{enc}(k_0)$ & $\text{enc}(k_1)$ & $\cdots$ & $\text{enc}(k_i)$ & $\cdots$ & $\text{enc}(k_n)$ \\
        |[datacell]| $\text{enc}(d_0)$ & |[datacell]| $\text{enc}(d_1)$ & |[datacell]| $\cdots$ & |[datacell]| $\text{enc}(d_i)$ & |[datacell]| $\cdots$ & |[datacell]| $\text{enc}(d_n)$ \\
    };
    \node[title, above=8mm of data] {Encrypted Data Array (Key-Value Pairs)};
    
    \node (query) [querycell, below=of data-2-4] {$\text{enc}(k_i)$};
    \node[left=2mm of query, font=\bfseries] {(Query)};
    
    \matrix (mask) [
        matrix of nodes, nodes={maskcell}, column sep=2mm, below=of query
    ] {
        $\text{enc}(\text{false})$ & $\text{enc}(\text{false})$ & $\cdots$ & $\text{enc}(\text{true})$ & $\cdots$ & $\text{enc}(\text{false})$ \\
    };
    \node[title, below=5mm of mask] {Resulting Encrypted Boolean Mask};
    
    \draw[arr, black, very thick] (query.north) -- (data-2-4.south);
    \node[align=left, font=\small, right=5mm of query] {
        HE Test:
        $\forall j \in [0,n]$,\\
        $\text{HE.eq}(\text{enc}(k_j), \text{query})$
    };
    \foreach \j in {1,2,4,6} {
        \draw[arr] (data-2-\j.south) -- (mask-1-\j.north);
    }
    \node[draw, red, dashed, inner sep=3pt, fit=(data-1-4) (data-2-4)] {};
    \node[draw, red, dashed, inner sep=3pt, fit=(mask-1-4)] {};
    
    
    \end{tikzpicture}
    \caption{Diagram of a homomorphic key-value search, generating an encrypted selection mask.}
    \label{fig:fhe_kv_search}
\end{figure*}
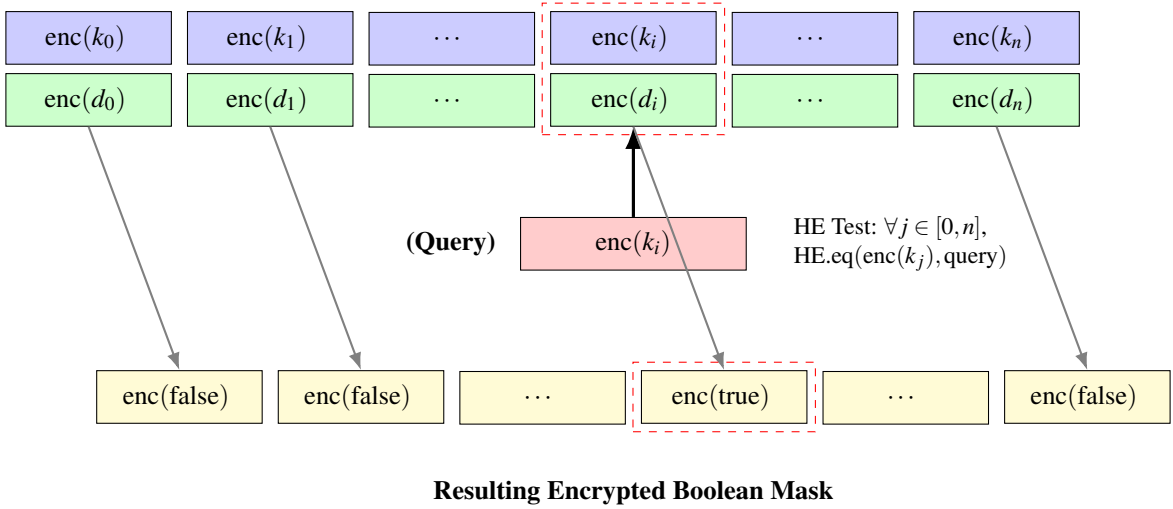

\subsubsection*{A. Data Access Control}
Our system enables data access control, ensuring only authorized users can access data. Each data owner has a unique keypair $(sk_{owner}, pk_{owner})$, with private keys staying exclusively on client devices. When storing data, we encapsulate it in an authenticated structure containing the encrypted payload, owner identifier, timestamp, and access metadata. This creates verifiable data ownership without requiring server access to private keys.

Data owners can delegate fine-grained permissions "Read", "Write", "Delete", and "Delegate" to other users through cryptographically signed tokens. These tokens contain authorization data: $\text{Delegation} = \{\text{owner\_id}$, $\text{user\_id}$, $\text{permissions}$, $\text{expires\_at}$, $\sigma_{owner}$, $Nonce\}$, where $\sigma_{owner}$ is the owner's digital signature. All database operations require cryptographically signed authentication tokens. The server verifies token authenticity through signature validation, ensuring only legitimate token holders can perform authorized operations. This zero-trust architecture prevents server access to encrypted data without proper authorization. We can substitute the signature scheme with a post-quantum one.

\subsubsection*{B. Query Processing}

The initial step in query execution involves parsing the input SQL statement into an Abstract Syntax Tree (AST). This tree serves as a structured, intermediate representation of the query, with distinct nodes corresponding to high-level clauses such as \texttt{SELECT}, \texttt{FROM}, and \texttt{WHERE}.

A critical phase of this process is the transformation of logical and comparison expressions, primarily within the \texttt{WHERE} clause, into a structure amenable to cryptographic evaluation. We term this structure a \textbf{Homomorphic Expression Tree}. In this tree, the leaf nodes represent operands---such as table identifiers and literal values---while the internal nodes represent operators that correspond to executable homomorphic operations. This architecture allows any binary expression supported by the underlying FHE library to be seamlessly integrated into the query plan. Supported homomorphic operations include not only but:

\begin{itemize}
    \item \textbf{Equality:} $a = b$
    \item \textbf{Comparison:} $a < b$ and $a \leq b$
    \item \textbf{Selection:} $\max(a, b)$ and $\min(a, b)$
\end{itemize}

\noindent To illustrate this transformation, consider the following SQL query, which filters records based on age and salary:

\noindent\begin{verbatim}[breaklines, fontsize=\small]{sql}
  SELECT * FROM table WHERE age > 18 AND salary < 1000;
\end{verbatim}

The parser converts this statement into a structured AST, represented below in JSON format. This AST precisely captures the logical hierarchy of the operations. The root of the \texttt{where} clause is a \texttt{BinaryExpression} with the operator \texttt{AND}, and its children are the two \texttt{ComparisonExpression} nodes.

\begin{verbatim}[breaklines, fontsize=\small]{json}
{
  "type": "SelectStatement",
  "columns": ["*"],
  "from": {"type": "Table", "name": "table"},
  "where": {
    "type": "BinaryExpression",
    "operator": "AND",
    "left": {
      "type": "ComparisonExpression",
      "operator": ">",
      "left": {"type": "Identifier", "name": "age"},
      "right": {"type": "Number", "value": 18}
    },
    "right": {
      "type": "ComparisonExpression",
      "operator": "<",
      "left": {"type": "Identifier", "name": "salary"},
      "right": {"type": "Number", "value": 1000}
    }
  }
}
\end{verbatim}

A key principle of this design is that operands within the tree (i.e., identifiers and numerical values) can be represented as either plaintext or FHE ciphertexts. However, the operators remain as plaintext metadata, which directs the server to apply the correct homomorphic functions during the evaluation phase.

Check algorithm \ref{alg:fhe_sql_client_2e} and algorithm \ref{alg:fhe_sql_server_2e} below for the client and server side algorithm details.

\section{Proof of Security for The PIR Property}
\label{sec:security}

We prove the security of the PIR property in our FHE-SQL protocol, denoted \prot{FHE-SQL}, within the Universal Composability (UC) framework \cite{UC}. The goal is to show that our real-world protocol emulates an ideal-world scenario where a trusted third party handles the query execution. This is achieved by proving that \prot{FHE-SQL} UC-realizes an ideal functionality, \func{FHE-SQL}, meaning that for any real-world adversary \adv, there exists a simulator \simul\ such that no environment \env\ can distinguish between an interaction with the real protocol and \adv, and an interaction with the \func{FHE-SQL} and \simul. We discuss only the single-server version of our protocol. The security of this may be compromised depending on the number of entries in the sharded database.

\begin{algorithm*}[t]
\caption{The Simulator $\simul$}
\label{alg:simulator}
\SetKwInOut{KwIn}{Input}
\SetKwInOut{KwOut}{Output}
\SetKwProg{Fn}{Procedure}{}{}
\SetKwFunction{Simulate}{Simulate}

\Fn{\Simulate{$pk, sk_{sig}$}}{
    Receive the leakage tuple $\leak = (Q_{redacted}, \text{Token}, DB)$ from \func{FHE-SQL}\;

    \tcp{Simulate Client Request:}
    Construct a fake encrypted query $Q'$ from $Q_{redacted}$ by replacing every placeholder '$\bot$' with $\text{Encrypt}_{pk}(0)$\;
    Provide the pair $(Q', \text{Token})$ to the internal adversary \adv, simulating the message sent from the client\;

    \tcp{Simulate Server Execution (in response to \adv's inspection):}
    Generate a fake encrypted database $DB'_{enc}$ of the size of $entries$, filled with $\text{Encrypt}_{pk}(0)$\;
    Generate a fake mask $M_{enc}'$ of the size of $entries$, filled with random-looking valid encryptions of zero\;
    Generate a fake result $C'_{result}$ with the correct dimensions, filled with random-looking valid encryptions\;

    Forward any messages from the internal \adv\ to the environment \env\ and vice-versa\;
}
\end{algorithm*}

\subsection{The Ideal Functionality: \func{FHE-SQL}}
\noindent The ideal functionality  represents a perfect, incorruptible service for SQL queries. It interacts with a user and a server that is controlled by the adversary \adv and precisely defines the information that is allowed to be leaked. \func{FHE-SQL} is initialized with a encrypted database $DB$ by the data owners. \\

\noindent\textbf{Initialize}: Upon receiving a message $(Init,DB')$ from a set of data owner parties: Set its internal database ${DB := DB'}$.

\noindent\textbf{Query}: Upon receiving a message $(Q, Token)$ from a user party:
\begin{enumerate}
    \item \textbf{Define Leakage}: Construct a leakage tuple $\leak(Q, DB, Token)$ containing:
          \begin{itemize}
              \item The query's Abstract Syntax Tree (AST), $Q_{redacted}$, with all literal values replaced by a placeholder symbol '$\bot$'. This reveals query structure, selected columns, and the target table.
              \item A valid authentication token for ${owner\_id}$.
              \item The total number of rows in the queried table.
          \end{itemize}
    \item \textbf{Send Leakage}: Send \leak\ to the adversary \adv.
    \item \textbf{Execute Query}: Execute the encrypted SQL query $q$ in the $DB$ to get the encrypted result set $C_{result}$.
    \item \textbf{Send Result}: Send $C_{result}$ to the user.
\end{enumerate}

\subsection{The Real World Protocol and Adversary}
In the real world, there is no trusted party. The user and server execute the \prot{FHE-SQL} protocol as defined in Algorithms \ref{alg:fhe_sql_client_2e} and \ref{alg:fhe_sql_server_2e}. The adversary \adv\ is assumed to be a static, malicious entity that corrupts the server. This gives \adv\ access to all of the server's internal states and computations, including the received encrypted query $Q$, the token, the full encrypted database, and all intermediate homomorphic calculations like $M_{enc}$ and the final encrypted result $C_{result}$.

\subsection{The Simulator \simul}
To prove security, we must construct a simulator \simul\ that operates in the ideal world but can generate a transcript that is computationally indistinguishable from a real-world execution. The simulator \simul\ internally runs a copy of the adversary \adv\ and uses the leakage \leak\ from \func{FHE-SQL} to simulate \adv.

\subsection{Indistinguishability Argument}

\newtheorem{theorem}{Theorem}
\begin{theorem}
    The protocol \prot{FHE-SQL} UC-realizes the ideal functionality \func{FHE-SQL} in the presence of a static, malicious adversary corrupting the server, assuming the underlying FHE and signature schemes are secure.
\end{theorem}

\begin{proof}[Proof Sketch]
    We argue that no environment \env\ can distinguish the real and ideal worlds.

    \begin{enumerate}
        \item \textbf{The Adversary's View:} In the real world, \adv\ sees the true encrypted query $Q$ and the true encrypted intermediate states. In the ideal world, \adv\ (running inside \simul) sees the simulated query $Q'$ and simulated states. Due to the semantic security of the FHE scheme, an encryption of a real literal value is computationally indistinguishable from an encryption of zero. Therefore, $Q$ is indistinguishable from $Q'$, and all subsequent encrypted states that depend on it are also indistinguishable from their simulated counterparts. The token is identical in both worlds, and \adv\ and \simul are not able to decrypt any value. Thus, the adversary's entire view in the real world is computationally indistinguishable from its view in the ideal world.

        \item \textbf{The User's View:} In the real world, the user receives $C_{result}$, decrypts it, and filters the zero-rows to obtain the correct plaintext result $R_{plain}$. In the ideal world, the user receives $C_{ideal}$ from \func{FHE-SQL}, which is constructed to also yield $R_{plain}$ after decryption and filtering. The user's final output is therefore identical in both worlds.
    \end{enumerate}

    \noindent Since the adversary's view is computationally indistinguishable and the user's output is identical, the environment \env\ cannot distinguish between the real world and the simulated world except with negligible probability. This concludes the proof sketch.
\end{proof}

\section{Performance Evaluation}
\label{sec:performance}

We conducted performance benchmarks across multiple dimensions to assess the viability of FHE-SQL. Our evaluation focuses on two primary aspects: (i) the performance of the storage layer, including both cached and non-cached data as well as the efficiency of concurrent access, and (ii) the computational overhead of the homomorphic operations in our protocol when accelerated with GPUs. Finally, we provided an estimate of the query execution time as a function of the number of entries in the database. Table \ref{tab:exp_setup} shows our experiment setup.
\begin{table}[H]
    \centering
    \begin{tabular}{ll}
        \toprule
        \textbf{Component} & \textbf{Specification} \\
        \midrule
        CPU                & AMD Ryzen 9 7950X3D 16-Core \\
        ISA                & x86-64 \\
        Memory             & DDR5 128GB \\
        Storage            & Samsung 980 Pro NVMe SSD 2TB \\
        GPU                & NVIDIA RTX 4080 Super 16GB \\
        \midrule
        OS                 & Ubuntu 24.04.2 LTS \\
        Rust               & rustc 1.87.0 (17067e9ac 2025-05-09) \\
        Rust-RocksDB       & 0.23.0 \\
        TFHE-rs            & 0.11.0 \\
        \bottomrule
    \end{tabular}
    \caption{Experimental Setup}
    \label{tab:exp_setup}
\end{table}

\subsection{Benchmarking Methodology}

Since the size of FHE ciphertexts often reaches hundreds of KB, we designed a blob storage system for larger data storage.  To evaluate the performance characteristics of our blob storage and RocksDB-based database under realistic conditions, we designed a comprehensive benchmarking framework using the criterion crate in Rust.

\subsubsection*{A. Concurrent Access and Workload Simulation}

The benchmark framework utilizes a multithreading approach to simulate concurrent user activity. For read-intensive scenarios, a pool of 64 worker threads was used, while write-intensive scenarios used 16 concurrent threads. To model real-world access patterns where certain keys are significantly "hotter" than others, the workload simulation employed a \textbf{Skew Normal distribution} ($\xi=-1.0, \omega=10.0, \alpha=30.0$). This generates a realistic, non-uniform sequence of key accesses. Object sizes for testing were categorized into three distinct sizes to reflect typical usage: small (64KB), medium (256KB), and large (1MB), with constants defined for each to ensure consistent testing across benchmarks.

\subsubsection*{B. Benchmark Scenarios}

\noindent The evaluation was divided into several key benchmark scenarios:

\begin{itemize}
    \item \textbf{Write Throughput}: This test measured the performance of sequential writes for small, medium, and large data sizes to establish a baseline for ingestion performance.
    \item \textbf{Read Latency}: Both cold-read (cache cleared) and hot-read (cache primed) latencies were measured for all data sizes. For cold-read tests, a cache-clearing function was invoked to ensure that the measurements reflected disk I/O performance rather than in-memory cache hits.
    \item \textbf{Batch Operations}: The efficiency of batch-writing 100 records at a time was tested for both small and large data sizes to measure the performance gains from reducing transaction overhead.
    \item \textbf{Access Simulation}: This benchmark simulated a high-concurrency environment with 64 threads performing reads based on the Skew Normal access pattern. The system was pre-populated with 10,000 items to ensure the caches and storage structures were in a realistic state before the measurement began.
\end{itemize}

\noindent While the current framework focuses on these core metrics, it is designed to be extensible for future tests, including more complex environmental simulations like network jitter and packet loss.

\begin{table*}[htbp]
    \centering
    \begin{tabular}{|l|c|c|c|c|}
        \hline
        \textbf{Operation}      & \textbf{Entry Size} & \textbf{Blob Storage} & \textbf{RocksDB} & \textbf{Blob Improvement} \\
        \hline
        \multicolumn{5}{|c|}{\textbf{Write Operations}} \\
        \hline
        Sequential Write Small  & 64KB                & 46.323 µs             & 120.85 µs        & 2.6x \\
        Sequential Write Medium & 256KB               & 184.46 µs             & 472.17 µs        & 2.6x \\
        Sequential Write Large  & 1MB                 & 739.58 µs             & 1.747 ms         & 2.4x \\
        Batch Write (100)       & 64KB                & 5.2498 ms             & 12.260 ms        & 2.3x \\
        \hline
        \multicolumn{5}{|c|}{\textbf{Read Operations}} \\
        \hline
        Cold Read Small         & 64KB                & 11.664 µs             & 5.5070 µs        & N/A \\
        Hot Read Small          & 64KB                & 1.1086 µs             & 2.4558 µs        & 2.2x \\
        Cold Read Medium        & 256KB               & 45.579 µs             & 18.709 µs        & N/A \\
        Hot Read Medium         & 256KB               & 4.4549 µs             & 7.3069 µs        & 1.6x \\
        Cold Read Large         & 1MB                 & 181.55 µs             & 73.686 µs        & N/A \\
        Hot Read Large          & 1MB                 & 15.796 µs             & 30.066 µs        & 1.9x \\
        Batch Read (100)        & 64KB                & 157.19 µs             & 321.14 µs        & 2.0x \\
        \hline
        \multicolumn{5}{|c|}{\textbf{Concurrent Operations}} \\
        \hline
        Concurrent Writes       & 64KB                & 131.33 µs             & 1.6872 ms        & 12.8x \\
        Concurrent Reads        & 64KB                & 39.372 µs             & 1.7871 ms        & 45.4x \\
        \hline
    \end{tabular}
    \caption{Performance Comparison: Blob Storage vs RocksDB}
    \label{tab:blob_rocksdb_comparison}
\end{table*}

\subsubsection*{A. Write Performance}
For single-threaded write operations, the blob storage solution consistently outperformed RocksDB across all object sizes. It was two times more faster for small sequential writes (46.3\,\textmu s vs. 120.8\,\textmu s) and batch writes (5.2\,ms vs. 12.2\,ms). This suggests that for simple data ingestion or bulk loading tasks without concurrency, blob storage offers superior throughput. In the scenario of concurrent write loads, blob storage can be 12x faster than RocksDB (131.33 \,\textmu s vs. 1.6872\,ms), which is a impressing result.

\subsubsection*{B. Read Performance}
\noindent Read performance highlights the most critical trade-offs, particularly regarding caching and data access patterns. In our benchmark, we fix the cache size of both system.

\begin{itemize}
    \item \textbf{Hot Reads (Cached Data):} Our hypothesis is strongly supported here. For data that is already in cache (``hot reads''), blob storage is consistently faster, showing roughly 2x better latency than RocksDB across all object sizes. This demonstrates a highly effective and efficient caching strategy within the blob storage. Under concurrent read pressure, this advantage holds, with blob storage being significantly faster (39.372\,\textmu s vs. 1.78\,ms).

    \item \textbf{Cold Reads (Uncached Data):} The performance profile inverts when the data are uncached. RocksDB is around 2x faster than blob storage. For a small 64KB object, RocksDB's cold read latency was 5.5070 \,\textmu s, while blob storage took 11.664 \,\textmu s. This suggests that a cache miss in the blob storage system likely incurs a slight penalty.
\end{itemize}

\subsubsection*{C. Benchmarks Observation}
The two approaches are optimized for different access patterns. If we use larger cache and memory, blob will outperform RocksDb in practice. Blob Storage is the ideal choice for read-intensive workloads with high data locality. RocksDB has low cold-read latency that makes it the clear winner for applications where performance must remain predictable and fast, even for randomly accessed, non-cached data.

The optimal storage strategy depends on the expected I/O patterns. Blob storage excels at serving cached or frequently accessed data, while RocksDB provides superior performance for low-latency random reads. FHE-SQL features an adaptive architecture that can leverage both. For mixed workloads with high read-write concurrency, it defaults to a hybrid model—storing large ciphertexts in blob storage and their metadata in RocksDB—thereby capitalizing on the latter's extremely low latency (typically ~1 µs) for rapid metadata lookups.

\subsubsection*{D. FHE Operation Impact on System Performance}
\noindent To perform a naive key-value lookup, a sequence of homomorphic operations is executed. This protocol assumes the database contains $n$ key-value pairs, stored in the form of $(\text{enc}(\text{u32}), \text{enc}(\text{u32}))$. The computational cost of a single query is composed of the following four stages:

\begin{enumerate}
    \item \textbf{Client-Side Query Encryption:} The client encrypts the query key $k_q$ to produce a single ciphertext, $\text{enc}(k_q)$. This constitutes one FHE encryption operation.
    \item \textbf{Server-Side Homomorphic Comparison:} The server iterates through all $n$ encrypted keys in the target column family. For each key $\text{enc}(k_i)$, it performs a homomorphic equality comparison against the encrypted query key.
          \begin{equation}
              \text{mask}_i = \text{HE.eq}(\text{enc}(k_i), \text{enc}(k_q)), \quad \forall i \in [1, n]
          \end{equation}
          This stage requires $n$ homomorphic equality operations and results in an encrypted boolean mask which is a vector of length $n$.
    \item \textbf{Server-Side Homomorphic Selection:} The server uses the boolean mask to isolate the target key-value pair. For each entry $i$, a homomorphic multiplexer (an \texttt{if-then-else} operation) selects either the stored ciphertext if $\text{mask}_i$ is true, or an encrypted zero otherwise. This is applied to both the key and value columns and requires $2n$ homomorphic multiplexer operations.
          \begin{align}
              \text{selected\_k}_i & = \text{HE.cmux}(\text{mask}_i, \text{enc}(k_i), \text{enc}(0)) \\
              \text{selected\_v}_i & = \text{HE.cmux}(\text{mask}_i, \text{enc}(v_i), \text{enc}(0)) \\
          \end{align}
          The resulting sparse vectors are sent back to client for decryption or can be then aggregated into a single key-value pair via homomorphic addition if the entry is unique. (But it requires extra homomorphic additions.)
          \begin{align}
              \text{result}_k & = \sum_{i=1}^{n} \text{selected\_k}_i \\
              \text{result}_v & = \sum_{i=1}^{n} \text{selected\_v}_i
          \end{align}
    \item \textbf{Client-Side Result Decryption:} The server either returns a sparse vector of encrypted values or the two final ciphertexts. Here, we let the server return a sparse vector of length n. The client performs decryption and extracts the values in plaintext. This requires $2n$ decryption operations.
\end{enumerate}

\begin{table}[ht]
\centering
\label{tab:fhe_op_performance}
\begin{tabular}{|l|c|}
\hline
\textbf{Operations} & \textbf{Median Latency} \\
\hline
\multicolumn{2}{|c|}{\textbf{Encryption \& Decryption}} \\
\hline
Encryption (u8) & 376.56 µs \\
Encryption (u32) & 1.4771 ms \\
Trivial Encryption (u8) &  585.88 ns \\
Trivial Encryption (u32) & 2.44 µs \\
Decryption (u8) & 2.1153 µs \\
Decryption (u32) & 8.4904 µs \\
\hline
\multicolumn{2}{|c|}{\textbf{Arithmetic Operations}} \\
\hline
Add (u8) & 32.725 ms \\
Add (u32) & 79.594 ms \\
\hline
\multicolumn{2}{|c|}{\textbf{Logical \& Comparison Operations}} \\
\hline
AND (u8) & 11.30 ms \\
AND (u32) & 23.92 ms \\
Equality (u8) & 17.62 ms \\
Equality (u32) & 41.91 ms \\
\hline
\multicolumn{2}{|c|}{\textbf{Control Flow}} \\
\hline
If-Then-Else (u8, trivial u8) & 22.37 ms \\
If-Then-Else (u32, trivial u32) & 48.91 ms \\
If-Then-Else (u32, u32) & 90.07 ms \\
\hline
\end{tabular}
\caption{Performance of Homomorphic Operations}
\end{table}
The costs of the naive key-value retrieval protocol are divided between the client and the server, with the server bearing the vast majority of the computational load. The total time is a function of $n$, the number of key-value pairs in the database. We still let the server return a sparse vector of length n, and all numbers are in u32. See Table 3 for the exact operation time. Even though we have tried to optimized the perform of data accessing, the real bottleneck of the entire system is the FHE computation overhead.
\label{sec:naive_retrieval}

\paragraph{Client-Side Computational Cost:}
The client's workload is constant and independent of the database size. It consists of one encryption and $2n$ decryptions.
The total client-side time is calculated as follows:
\begin{align*}
    T_{client} & = 1 \times (1.4771 \,\text{ms}) + 2n \times (8.4904 \,\mu\text{s}) \\
               & \approx (1.4771+0.017\cdot n) \,\text{ms}
\end{align*}

\paragraph{Server-Side Computational Cost:}
The server's workload scales linearly with the number of entries, $n$. It is dominated by two homomorphic computations:
\begin{itemize}
    \item $n$ homomorphic equality tests (\texttt{eq\_u32})
    \item $2n$ homomorphic multiplexers (\texttt{if\_then\_else})
\end{itemize}

Using the median latencies from the benchmarks presented in Table 3, and we ignore trivial encryption. The total server-side time is calculated as follows:
\begin{align*}
    T_{server} & = n \cdot T_{eq\_u32} + 2n \cdot T_{if\_then\_else}             \\
               & = n \times (41.91\,\text{ms}) + 2n \times (48.91\,\text{ms})    \\
               & = 41.91\cdot n + 97.82\cdot n \approx 139.73\cdot n \,\text{ms} \\
\end{align*}
For any non-trivial database size, the server-side latency is approximately \textbf{139.73 milliseconds per entry}. We can leverage database sharding to parallelize the workload of data sql selection. When we execute an SQL query, the latency can increase severalfold, depending on the commands used.

\section{Discussion}
\label{sec:discussion}

The development and evaluation of FHE-SQL have yielded valuable insights into the practical application of fully homomorphic encryption to database systems. Our work successfully demonstrates the feasibility of building a privacy-preserving SQL interface, but it also illuminates the significant challenges and trade-offs inherent in this approach. This section discusses the primary performance bottlenecks, the fundamental limitations of the current design, and promising directions for future research.

\subsection{System Limitations}

\subsubsection*{Primacy of Homomorphic Computation}
The performance analysis in Section~\ref{sec:performance} unequivocally demonstrates that the computational overhead of homomorphic operations is the dominant bottleneck in the FHE-SQL system. With individual server-side FHE operations requiring tens of milliseconds, the microsecond-level latency of the underlying storage access---even for cold reads from our optimized blob storage---becomes practically negligible in the overall query execution time. As calculated, a naive query costs approximately $139.73 \cdot n$ milliseconds, a cost entirely dictated by the cryptographic computations.

This disparity implies that for the foreseeable future, optimizations at the storage layer will yield diminishing returns. Efforts to further reduce read/write latency, while beneficial in traditional database systems, are not a priority for FHE-SQL. The critical path to performance improvement lies in accelerating the homomorphic computations themselves, either through algorithmic advances in FHE schemes, hardware acceleration, or architectural optimizations at the query processing layer. Until the execution time of FHE operations approaches the same order of magnitude as storage I/O, the latter can be largely ignored in performance modeling.

\subsubsection*{The Inevitability of Linear Scans}
A core limitation stemming directly from the use of semantic-level encryption is the inability to build traditional database indexes. Structures like B-Trees rely on plaintext comparisons to efficiently navigate and prune the search space. Since all data in FHE-SQL remains encrypted, the server cannot perform these comparisons. Any data structure that reveals ordering or equality information would violate the privacy guarantees that FHE is designed to provide.

Consequently, FHE-SQL is constrained to a "full scan" or linear iteration approach for query processing. To find a matching record, the server must homomorphically compare the encrypted query predicate against every entry in the target column family. This results in a computational complexity of $O(n)$ for a simple key-value lookup, which poses a significant barrier to scalability for large datasets. While our system provides strong privacy, it trades the sub-linear query time of conventional databases for a protocol that is computationally intensive but reveals no information about which records are being accessed.

\section{Conclusion}
\label{sec:conclusion}
We presented FHE-SQL, a system that addresses security flaws in property-preserving encryption by enabling SQL queries directly on encrypted data with fully homomorphic encryption. We presented an architecture featuring an oblivious query protocol and a hybrid storage system to manage large ciphertexts, with security proven under the Universal Composability framework.

Our benchmarks confirmed that the primary factor influencing performance is the computational overhead of FHE. Consequently, future work can most effectively enhance scalability by focusing on horizontal scaling via database sharding, parallel query execution, and the design of FHE-aware optimizers.

In conclusion, FHE-SQL provides a blueprint for the future of truly private database systems. This work establishes a clear and promising path toward building the next generation of secure cloud services, paving the way for the continued maturation and widespread adoption of FHE technology.

\bibliographystyle{plain}
\bibliography{reference.bib}

\end{document}